\documentclass[11pt]{article}
\usepackage[height=220mm,width=150mm]{geometry}
\usepackage{array}
\usepackage{color}
\usepackage{mathrsfs}
\usepackage{amssymb}
\usepackage{amsmath}
\usepackage{amsthm}
\allowdisplaybreaks
\usepackage{graphicx}
\usepackage{tikz}
\usetikzlibrary{shapes.arrows,chains,positioning}
\newtheorem{Theorem}{Theorem}[section]
\newtheorem{Lemma}[Theorem]{Lemma}
\newtheorem{Definition}[Theorem]{Definition}
\newtheorem{Proposition}[Theorem]{Proposition}
\newtheorem{Corollary}[Theorem]{Corollary}
\newtheorem{Remark}[Theorem]{Remark}
\usepackage{epsfig,psfrag}

\begin{document}

\title{A new convolution theorem associated with the linear canonical transform}

%\titlerunning{Short form of title}        % if too long for running head
\author{Haiye Huo\\
Department of Mathematics, Nanchang University, \\Nanchang~330031, China\\
\mbox{} \\
Email:  hyhuo@ncu.edu.cn}

\date{}
\maketitle

\textit{Abstract}.\,\,
In this paper, we first introduce a new notion of canonical convolution operator, and show that it satisfies the commutative, associative, and distributive properties, which may be quite useful in signal processing.
Moreover, it is proved that the generalized convolution theorem and generalized Young's inequality are also hold for the new canonical convolution operator associated with the LCT. Finally, we investigate the sufficient and necessary conditions for solving a class of convolution equations associated with the LCT.


\section{Introduction}\label{sec I}
The linear canonical transform (LCT) \cite{BKO1997,Bernardo1996,MQ1971,WWL2016} is a class of linear integral transform characterized with parameter $A=(a,b,c,d)$. It is well known that Fourier transform, Fresnel transform, fractional Fourier transform \cite{Wei2016}, and scaling operations, are all special cases of the LCT by choosing specific parameters of $A$. Therefore, the LCT has recently drawn much attention as a powerful mathematical tool in the fields of signal processing, communications and optics \cite{OZK2001,QLL2013}.

So far, many classical results in the Fourier transform domain have been extended to the LCT domain, for instance, sampling theorem \cite{HS2015,SSS2015,WL2014,WL2016,XS2013,XTZ2017}, uncertainty principle \cite{HZCX2016,SHZ2016,Stern2008,Zhang2016,ZTLW2009}, convolution theorem \cite{DTW2006,PD2001,SLZ2014B,WRL2012,XQ2014}, etc. In this paper, we revisit the convolution theorem for the LCT. As we all known, the classical convolution theorem for the Fourier transform states that the Fourier transform of the convolution of two functions is equal to the pointwise product of the Fourier transforms.
That is to say, in the Fourier transform domain, the classical convolution transform is expressed as:
\begin{equation}
\mathcal{F}(f*g)(u)=\mathcal{F}f(u)\mathcal{F}g(u),
\end{equation}
where the convolution operator $*$ is defined by
\begin{equation}\label{Fourier1}
f*g(t)=\int_{-\infty}^{+\infty}f(\tau)g(t-\tau){\rm d}\tau.
\end{equation}
Unfortunately, this claim is not true for the LCT. Therefore, by defining different forms of convolution operators (called canonical convolution operators in order to distinguish from the aforementioned convolution operator associated with the Fourier transform), a variety of convolution theorems for the LCT have been derived, see for example, Pei and Ding \cite{PD2001}, Deng et al. \cite{DTW2006}, Wei et al. \cite{WRL2012,WRL2009}, Shi et al. \cite{SLZ2014B,SSZZ2012}.

Pei and Ding \cite{PD2001} introduced a canonical convolution operator $O_{conv}^A$, which is denoted as
  \begin{align}
  &O_{conv}^A(f(t),g(t))=\mathcal{L}_{A^{-1}}\big\{\mathcal{L}_{A}f(u)\mathcal{L}_{A}g(u)\big\}(t)\nonumber\\
  =&\sqrt{\frac{1}{j8\pi^3b^3}}\int_{\mathbb{R}^3}e^{j\frac{a}{2b}(v^2+u^2+\tau^2-t^2)-j\frac{u}{b}(t-\tau-v)}\nonumber\\
    &\times f(v)g(\tau){\rm d}v {\rm d}u{\rm d}\tau.\label{Pei1}
  \end{align}
Then, the convolution theorem associated with the LCT can be expressed as follows
\begin{equation}
  \mathcal{L}_{A}\{O_{conv}^A(f(t),g(t))\}(u)=\mathcal{L}_{A}f(u)\mathcal{L}_{A}g(u),
  \end{equation}
which is similar to the traditional convolution theorem in the Fourier transform domain.
However, we can see from $(\ref{Pei1})$ that it is difficult to reduce $O_{conv}^A(f(t),g(t))$ into a single integral form as in the traditional convolution operator formula (\ref{Fourier1}).

Deng et al. \cite{DTW2006} proposed another canonical convolution operator $\Theta$, which is defined by
  \begin{small}
  \begin{align}
  (f\Theta g)(t)
  =&\sqrt{\frac{1}{j2\pi b}}e^{-j\frac{a}{2b}t^2}\bigg(\Big(e^{j\frac{a}{2b}t^2}f(t)\Big)*\Big(e^{j\frac{a}{2b}t^2} g(t)\Big)\bigg)\nonumber\\
  =&\int_{-\infty}^{+\infty}f(\tau)g(t-\tau)e^{-j\frac{a}{b}\tau(t-\tau)}{\rm d}\tau.\label{DW1}
  \end{align}
  \end{small}
Thus, the convolution theorem in the LCT domain can be represented as
\begin{equation}\label{eq:Q1}
   \mathcal{L}_{A}(f\Theta g)(u)=\mathcal{L}_{A}f(u)\mathcal{L}_{A}g(u)e^{-j\frac{d}{2b}u^2}.
  \end{equation}
Later, Wei et al. \cite{WRL2012} independently investigated the convolution theorem (\ref{eq:Q1}) and obtained some extended results.

Moreover, Wei et al. \cite{WRL2009} proposed a new canonical convolution operator $\overset{A}{\Theta}$ as follows
  \begin{equation}\label{Wei1}
  \Big(f\overset{A}{\Theta} g\Big)(t)=\int_{-\infty}^{+\infty}f(\tau)g(t\theta\tau){\rm{d}}\tau.
  \end{equation}
  Here, $g(t\theta\tau)$ is the $\tau$-generalized translation of $g(t)$, which is defined by
  \begin{align}
  g(t\theta\tau)
  =&\sqrt{\frac{1}{j2\pi b}}\sqrt{\frac{1}{-jb}}e^{-j\frac{a}{2b}(t^2-\tau^2)}\nonumber\\
     &\times\frac{1}{\sqrt{2\pi}}\int_{-\infty}^{+\infty}\mathcal{L}_{A}g(u)e^{j\frac{1}{b}(t-\tau)u}{\rm{d}}u.\label{theta}
  \end{align}
  Hence, the convolution theorem associated with the LCT becomes
  \begin{equation}
   \mathcal{L}_{A}\Big(f\overset{A}{\Theta} g\Big)(u)=\mathcal{L}_{A}f(u)\mathcal{L}_{A}g(u).
  \end{equation}
  The form (\ref{Wei1}) is also quite simple with respect to that of Fourier transform.

 Furthermore, Shi et al. \cite{SSZZ2012} introduced a new canonical convolution structure for the LCT,
 and the canonical convolution operator $\Theta_M$ is denoted as
  \begin{equation}\label{Shi11}
  (f\Theta_M g)(t)=\int_{-\infty}^{+\infty}f(\tau)g(t-\tau)e^{-j\frac{a}{b}\tau(t-\frac{\tau}{2})}{\rm d}\tau.
  \end{equation}
 Therefore, the convolution theorem has the following form
 \begin{equation}
  \mathcal{L}_{A}(f\Theta_M g)(u)=\sqrt{2\pi}\mathcal{L}_{A}f(u)\mathcal{F}g\big(\frac{u}{b}\big).
  \end{equation}
 It is shown in \cite{SSZZ2012} that the new canonical convolution operator is quite useful for signal processing.

Later, Shi et al. \cite{SLZ2014B} proposed another canonical convolution operator in the following
 \begin{align}
  &(f\Xi_{A_1,A_2,A_3} g)(t)\nonumber\\
  =&\int_{-\infty}^{+\infty}(\mathrm{T}_{\tau}^{A_1}f)(t)g(\tau)\rho_{a_1,a_2,a_3}(t,\tau){\rm d}\tau.\label{Shi22}
  \end{align}
  Here,
  \begin{equation}\label{Tau}
  (\mathrm{T}_{\tau}^{A_1}f)(t)=f(t-\tau)e^{-j\frac{a_1}{b_1}\tau(t-\frac{\tau}{2})},
  \end{equation}
  and
  \[
  \rho_{a_1,a_2,a_3}(t,\tau)=e^{j\frac{a_2}{b_2}\tau^2+j\big(\frac{a_1}{2b_1}-\frac{a_3}{2b_3}\big)t^2}.
  \]
 Then, the convolution theorem for the LCT has the form
 \begin{align}
    &\mathcal{L}_{A_3}(f\Xi_{A_1,A_2,A_3} g)(u)\nonumber\\
    =&\epsilon_{d_1,d_2,d_3}(u)\mathcal{L}_{A_1}\Big(\frac{b_1}{b_3}u\Big)\mathcal{L}_{A_2}\Big(\frac{b_2}{b_3}u\Big),\label{ShiA}
 \end{align}
 where
   \[
  \epsilon_{d_1,d_2,d_3}(u)=\sqrt{\frac{j2\pi b_1 b_2}{b_3}}e^{ju^2\Big(\frac{d_3}{2b_3}-\frac{d_1b_1^2}{2b_1b_3^2}-\frac{d_2b_2^2}{2b_2b_3^2}\Big)}.
  \]
It follows from \cite{SLZ2014B} that the classical convolution theorem for the Fourier transform, the generalized convolution theorem for the fractional Fourier transform, and some existing canonical convolution theorems associated with the LCT can be regarded as the special cases for (\ref{ShiA}).

The canonical convolution operators for the LCT introduced in the six papers mentioned above are very interesting, and can be applied to solving many theoretical or practical problems, since they can be considered as some extensions of classical convolution operator for the Fourier transform.
%However, there are only few references focus on studying the canonical convolution theorems associated with the LCT, compared to those of Fourier transform. From my perspective, the reason is that the LCT is much more complicated than the Fourier transform.
%
%\textcolor[rgb]{1.00,0.00,0.00}{The motivation of the paper is quite similar to those of six papers listed above (i.e., \cite{DTW2006,PD2001,SLZ2014B,SSZZ2012,WRL2012,WRL2009}).
In this paper, our goal is to introduce a new canonical convolution operator for the LCT, and then derive a generalized version of the classical convolution theorem, and Young's inequality associated with the Fourier transform. Furthermore, we discuss the solvability of a class of convolution equations associated with the new canonical convolution operator. In Sec~\ref{sec:New}, we verify that our new defined canonical convolution operator can be performed into two different ways for implement in filter design. This fact may have some advantages over others in filter design, for example, compared with the canonical convolution operators introduced in references \cite{DTW2006,SLZ2014B,SSZZ2012,WRL2012,WRL2009}, which only have one way of convoluting, respectively. In fact, by considering the computational complexity or input conditions, we can have two options for choosing filtering, since in some cases, the first one may be perform better than the second one, or vice versa. Therefore, when applied to solving some specific problems, our canonical convolution introduced in this paper is much more flexible than the existing ones for the LCT mentioned in \cite{DTW2006,SLZ2014B,SSZZ2012,WRL2012,WRL2009}.

%However, to the best of our knowledge, there are only few results on unifying the formulae of the convolution theorems for the Fourier transform, fractional Fourier transform, and LCT. Hence, in this paper, our main purpose is to introduce a new convolution theorem for the LCT, which includes the classical convolution theorems for the Fourier transform and fractional Fourier transform as special cases.
%Moreover, our new defined canonical convolution operator only includes a simple integral and it can be performed into two different ways for implement in filter design. Hence, it is much more flexible and useful than the existing ones for the LCT mentioned in \cite{DTW2006,PD2001,SLZ2014B,SSZZ2012,WRL2012,WRL2009,XQ2014}.

%The main purpose of this paper is defining a new canonical convolution operator associated with the LCT. Then, we present the generalized version of the convolution theorem, and Young's inequality associated with the LCT. Finally, we give a unique solution for a class of convolution equations associated with the new canonical convolution operator.

The rest of paper is organized as follows. In Sec~\ref{sec:Prel}, we briefly recall the definition of the LCT.
In Sec~\ref{sec:New}, we introduce a new canonical convolution operator, and prove that it satisfies the generalized convolution theorem for the LCT. In Sec~\ref{sec:Some}, we present two applications for the new canonical convolution operator. First, we derive a generalized Young's inequality for the new canonical convolution operator associated with the LCT. Second, we give some sufficient and necessary conditions for the solvability of a class of convolution equations associated with our new defined canonical convolution operator. Finally, we conclude the paper.

%\section{Preliminaries}\label{sec:Prel}
\section{The Linear Canonical Transform}\label{sec:Prel}
%\subsection{The Linear Canonical Transform}

\begin{Definition}
The LCT of a signal $f(t)\in L^1(\mathbb{R})$ is defined by \cite{Stern2006B}:
\begin{small}
\begin{align}
&\mathcal{L}_{A}f(u):=\mathcal{L}_{A}\{f(t)\}(u)\nonumber\\
=&
    \begin{cases}
    \displaystyle{\sqrt{\frac{1}{j2\pi b}}\int_{-\infty}^{+\infty}f(t)e^{j\frac{a}{2b}t^{2}-j\frac{1}{b}ut+j\frac{d}{2b}u^{2}}{\rm{d}}t}, & b\ne0,\\
     \sqrt{d}e^{j\frac{cd}{2}u^{2}}f(du), & b=0.\\
     \end{cases}\label{def:LCT:L0}
\end{align}
\end{small}
where $A=(a,b,c,d)$,
%\left(
%\begin{array}{cc}
%a & b \\
%c & d \\
%\end{array}
%\right)$,
and parameters $a,\;b,\;c,\;d\in \mathbb{R}$ satisfy $ad-bc=1$.
\end{Definition}

For $b=0$, the LCT becomes a Chirp multiplication operator. Hence, without loss of generality, we assume that $b\neq0$ in the rest of the paper. As aforementioned, the LCT includes many linear integral transforms as special cases. For instance, let $A=(0,1,-1,0),$ then the LCT (\ref{def:LCT:L0}) reduces to the Fourier transform; let $A=(\cos\alpha,\sin\alpha,-\sin\alpha,\cos\alpha),$ then the LCT (\ref{def:LCT:L0}) becomes to the fractional Fourier transform.

\section{A New Generalized Convolution Theorem for the LCT}\label{sec:New}

In this section, we first introduce a new canonical convolution operator which is quite different from the existing ones. It is shown that our new canonical convolution operator is much more flexible and useful in certain cases. Then, we study the corresponding generalized convolution theorem associated with the LCT. Finally, we give several properties that the new canonical convolution operator satisfies.

First, we introduce a new notion of canonical convolution operator, which is related to the LCT parameter $A$. Our new definition is a generalized version of \cite[Definition~1]{ACTT2017A}.
\begin{Definition}\label{def:convo}
Given two functions $f,\;g\in L^1(\mathbb{R})$, the canonical convolution operator $\otimes_A$ is denoted as
\begin{eqnarray}
(f\otimes_A g)(t)&=&\sqrt{\frac{1}{j2\pi b}}\int_{-\infty}^{+\infty}f(u)g(t-u+b){}\nonumber\\
{}&&\times e^{j\frac{a}{b}u^{2}-j\frac{a}{b}ut+jat-jau}{\rm{d}}u,\label{def:convo:1}
\end{eqnarray}
where $A=(a,b,c,d)$ is defined the same as the LCT parameter.
\end{Definition}

The new canonical convolution expression (\ref{def:convo:1}) can be rewritten into two different forms, according to the classical convolution operator $*$. First, it can be represented as
\begin{align}
h(t):=&(f\otimes_A g)(t)\nonumber\\
=&(e^{j\frac{a}{2b}s^2}\cdot f(s))*(e^{jas}\cdot e^{j\frac{a}{2b}s^2}\cdot g(s+b))(t)\nonumber\\
&\times \sqrt{\frac{1}{j2\pi b}}e^{-j\frac{a}{2b}t^2}.\label{relat:1}
\end{align}
Second, it can also be described as
\begin{align}
h(t):=&(f\otimes_A g)(t)\nonumber\\
=&(e^{-jas}\cdot e^{-j\frac{a}{2b}s^2}\cdot f(s))*(e^{j\frac{a}{2b}s^2}\cdot g(s+b))(t)\nonumber\\
&\times \sqrt{\frac{1}{j2\pi b}}e^{j\frac{a}{2b}t^2+jat}.\label{relat:2}
\end{align}
Thanks to (\ref{relat:1}) and (\ref{relat:2}), we give two realizations for the new canonical convolution operator $\otimes_A$ in Fig.~\ref{Fig.1}, and Fig.~\ref{Fig.2}, respectively.

% For two-column wide figures use
\begin{figure*}
% Use the relevant command to insert your figure file.
% For example, with the graphicx package use
  \includegraphics[width=0.75\textwidth]{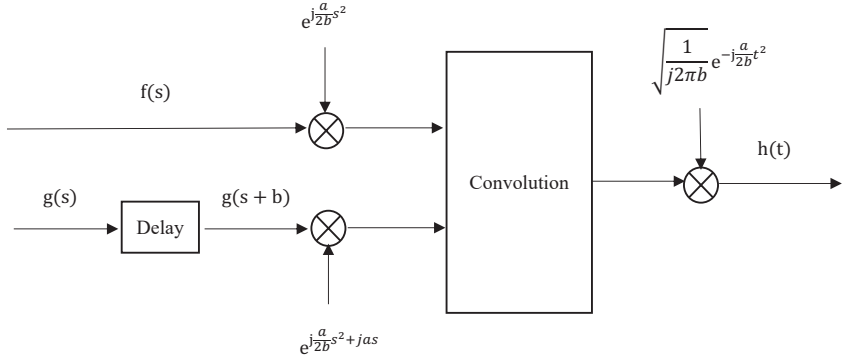}
% figure caption is below the figure
  \vspace{-20mm}
\caption{One expression of the canonical convolution operator $\otimes_A$}
\label{Fig.1}       % Give a unique label
\end{figure*}

\begin{figure*}
% Use the relevant command to insert your figure file.
% For example, with the graphicx package use
  \includegraphics[width=0.76\textwidth]{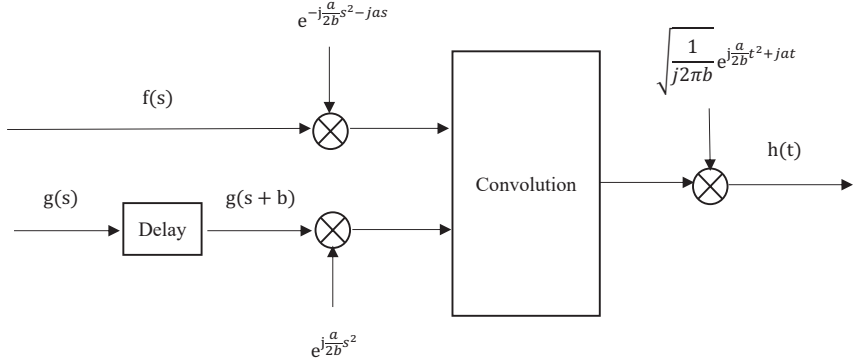}
% figure caption is below the figure
%\vspace{-15mm}
\caption{Alternative representation of the canonical convolution operator $\otimes_A$}
\label{Fig.2}       % Give a unique label
\end{figure*}

Compared to the canonical convolution operator defined in (\ref{Pei1}), our new canonical convolution is a single integral, which is much simpler than the triple integral mentioned in (\ref{Pei1}). Furthermore, the definition of canonical convolution operator in (\ref{Shi22}) is so complicated that it is not quite useful in filter design. Although the form of the new defined canonical convolution operator (\ref{def:convo:1}) is similar to those of (\ref{DW1}), (\ref{Wei1}) and (\ref{Shi11}), it follows from (\ref{relat:1}) and (\ref{relat:2}) that the new canonical convolution operator $\otimes_A$ has two realizations in filer design. Therefore, in certain cases, our new canonical convolution operator is much more flexible and useful than those in \cite{DTW2006,PD2001,SLZ2014B,SSZZ2012,WRL2012,WRL2009}.

Based on the new canonical convolution operator $\otimes_A$, we derive the generalized convolution theorem associated with the LCT as follows.
\begin{Theorem}\label{Thm:A}
Let $f,\;g\in L^1(\mathbb{R})$,\;$\Phi(u):=e^{ju-j\frac{d}{2b}u^2-j\frac{ab}{2}}$. Then, we have
\begin{equation}\label{ThmA:1}
\|f\otimes_A g\|_1\le\sqrt{\frac{1}{2\pi|b|}}\|f\|_1\|g\|_1,
\end{equation}
and
\begin{equation}\label{ThmA:2}
\mathcal{L}_{A}(f\otimes_A g)(u)=\Phi(u)\mathcal{L}_{A}f(u)\mathcal{L}_{A}g(u).
\end{equation}
\end{Theorem}

\begin{proof}
First, we prove (\ref{ThmA:1}).
Let $s=t-u+b$. By the definition of canonical convolution operator $\otimes_A$ (\ref{def:convo:1}), we have

\begin{eqnarray}
\|f\otimes_A g\|_1
&=&\int_{-\infty}^{+\infty}|(f\otimes_A g)(t)|{\rm{d}}t\nonumber\\
&\le&\sqrt{\frac{1}{2\pi|b|}}\int_{-\infty}^{+\infty}\int_{-\infty}^{+\infty}|f(u)g(t-u+b)|{\rm{d}}u{\rm{d}}t\nonumber\\
&=&\sqrt{\frac{1}{2\pi|b|}}\int_{-\infty}^{+\infty}|f(u)|{\rm{d}}u\int_{-\infty}^{+\infty}|g(s)|{\rm{d}}s\nonumber\\
&=&\sqrt{\frac{1}{2\pi|b|}}\|f\|_1\|g\|_1.\label{ThmA:4}
\end{eqnarray}
Next, we prove (\ref{ThmA:2}). By using the definition of the LCT, making change of variable $v=s-t+b$, and then utilizing the definition of canonical convolution operator (\ref{def:convo:1}), we have

\begin{small}
\begin{align}
&\Phi(u)\mathcal{L}_{A}f(u)\mathcal{L}_{A}g(u)\nonumber\\
=&e^{ju-j\frac{d}{2b}u^2-j\frac{ab}{2}}\displaystyle{\frac{1}{j2\pi b}\int_{-\infty}^{+\infty}f(t)e^{j\frac{a}{2b}t^{2}-j\frac{1}{b}ut+j\frac{d}{2b}u^{2}}{\rm{d}}t}\nonumber\\
    &\quad\times\displaystyle{\int_{-\infty}^{+\infty}g(v)e^{j\frac{a}{2b}v^{2}-j\frac{1}{b}uv+j\frac{d}{2b}u^{2}}{\rm{d}}v}\nonumber\\
=&e^{ju-j\frac{d}{2b}u^2-j\frac{ab}{2}}\frac{1}{j2\pi b}\int_{-\infty}^{+\infty}\int_{-\infty}^{+\infty}f(t)g(v)\nonumber\\
&\quad \times e^{j\frac{a}{2b}(t^{2}+v^2)-j\frac{1}{b}u(t+v)+j\frac{d}{b}u^{2}}{\rm{d}}t{\rm{d}}v\nonumber\\
=&\frac{1}{j2\pi b}e^{-j\frac{ab}{2}}\int_{-\infty}^{+\infty}\int_{-\infty}^{+\infty}f(t)g(v)\nonumber\\
    &\quad \times e^{j\frac{a}{2b}(t^{2}+v^2)-j\frac{1}{b}u(t+v-b)+j\frac{d}{2b}u^{2}}{\rm{d}}t{\rm{d}}v\nonumber\\
=&\frac{1}{j2\pi b}e^{-j\frac{ab}{2}}\int_{-\infty}^{+\infty}\int_{-\infty}^{+\infty}f(t)g(s-t+b)\nonumber\\
    &\quad \times e^{j\frac{a}{2b}[t^{2}+(s-t+b)^2]-j\frac{1}{b}us+j\frac{d}{2b}u^{2}}{\rm{d}}s{\rm{d}}t\nonumber\\
=&\frac{1}{j2\pi b}e^{-j\frac{ab}{2}}\int_{-\infty}^{+\infty}\int_{-\infty}^{+\infty}f(t)g(s-t+b)\nonumber\\
   &\quad\times e^{j\frac{a}{2b}[2t^{2}-2st+s^2+2bs-2bt+b^2]-j\frac{1}{b}us+j\frac{d}{2b}u^{2}}{\rm{d}}s{\rm{d}}t\nonumber\\
=&\sqrt{\frac{1}{j2\pi b}}\int_{-\infty}^{+\infty} e^{j\frac{a}{2b}s^{2}-j\frac{1}{b}us+j\frac{d}{2b}u^{2}}\nonumber\\
   &\;\times \left\{\sqrt{\frac{1}{j2\pi b}}\int_{-\infty}^{+\infty}f(t)g(s-t+b)
        e^{j\frac{a}{b}t^{2}-j\frac{a}{b}st+jas-jat}{\rm{d}}t\right\}{\rm{d}}s\nonumber\\
=&\sqrt{\frac{1}{j2\pi b}}\int_{-\infty}^{+\infty} e^{j\frac{a}{2b}s^{2}-j\frac{1}{b}us+j\frac{d}{2b}u^{2}}(f\otimes_A g)(s){\rm{d}}s\nonumber\\
=&\mathcal{L}_{A}(f\otimes_A g)(u).\label{ThmA:5}
\end{align}
\end{small}
This completes the proof.
\end{proof}

%Next, we show that our new convolution theorem is a more general form, which includes the classical convolution theorem for the Fourier transform \cite{Grochenig2001} and the generalized convolution theorem associated with the fractional Fourier transform \cite{ACTT2017A} as special cases.
%
%\begin{Corollary}
% Let $A=(0,1,-1,0),$ then Theorem~\ref{Thm:A} reduces to the classical convolution theorem for the Fourier transform mentioned in \cite{Grochenig2001}.
%\end{Corollary}
%
%\begin{Corollary}
% Let $A=(\cos\alpha,\sin\alpha,-\sin\alpha,\cos\alpha),$ then Theorem~\ref{Thm:A} reduces to the generalized convolution theorem for the fractional Fourier transform mentioned in \cite[Theorem 1]{ACTT2017A}.
%\end{Corollary}

After simple computation, it follows from Theorem~\ref{Thm:A} that the canonical convolution operator $\otimes_A$ satisfies three properties: Commutative property, associative property, and distributive property. More clearly, the following three equalities hold for any $f,\;g,\;h\in L^1(\mathbb{R})$:
\begin{enumerate}
  \item[(1)] Commutativity: $f\otimes_A g=g\otimes_A f.$
  \item[(2)] Associativity: $(f\otimes_A g)\otimes_A h=f\otimes_A(g\otimes_A h)$.
  \item[(3)] Distributivity: $f\otimes_A(g+h)=f\otimes_A g+f\otimes_A h$.
\end{enumerate}

\section{Two Applications for the New Canonical Convolution Operator}\label{sec:Some}

\subsection{Generalized Young's Inequality}\label{Sec:Generalized}

In this subsection, we investigate the generalized Young's inequality for the new canonical convolution operator $\otimes_A$. First, let us recall the classical Young's inequality as follows.

\begin{Proposition}[\cite{Grochenig2001}]\label{Prop:A}
Let $f\in L^p(\mathbb{R}),\; g\in L^q(\mathbb{R})$,\; $\frac{1}{p}+\frac{1}{q}=1+\frac{1}{r}$,\; $\frac{1}{r}+\frac{1}{r^{\prime}}=1$. Then,
\begin{equation}\label{Prop:A1}
\|f*g\|_r\le A_pA_qA_{r^{\prime}}\|f\|_p\|g\|_q,
\end{equation}
where
\begin{equation}\label{Prop:A2}
A_p=\Big({\frac{p^{1/p}}{{p^{\prime}}^{1/{p^{\prime}}}}}\Big)^{1/2},
\end{equation}
and
$1/p+1/p^{\prime}=1$.
\end{Proposition}

Next, we show that our new canonical convolution operator $\otimes_A$ also satisfies the Young's inequality.

\begin{Theorem}\label{Thm:C}
Let $f\in L^p(\mathbb{R}),\; g\in L^q(\mathbb{R})$,\; $\frac{1}{p}+\frac{1}{q}=1+\frac{1}{r}$,\;$\frac{1}{r}+\frac{1}{r^{\prime}}=1$. Then,
\begin{equation}\label{ThmC:1}
\|f\otimes_A g\|_r\le \sqrt{\frac{1}{2\pi |b|}}A_p A_q A_{r^{\prime}}
\|f\|_p\|g\|_q,
\end{equation}
where $A_p$ is defined the same as in (\ref{Prop:A2}).
\end{Theorem}

\begin{proof}
By (\ref{relat:1}), we obtain
\begin{small}
\begin{align}
&\|f\otimes_A g\|_r\nonumber\\
=&\bigg(\int_{-\infty}^{+\infty}\Big|(e^{j\frac{a}{2b}s^2}\cdot f(s))*(e^{jas}\cdot e^{j\frac{a}{2b}s^2}\cdot g(s+b))(t)\nonumber\\
&\times \sqrt{\frac{1}{j2\pi b}}e^{-j\frac{a}{2b}t^2}\Big|^r{\rm{d}}t\bigg)^{\frac{1}{r}}\nonumber\\
=& \sqrt{\frac{1}{2\pi |b|}}\bigg(\int_{-\infty}^{+\infty}\bigg|(e^{j\frac{a}{2b}s^2}f(s))*(e^{jas+j\frac{a}{2b}s^2} g(s+b))(t)\bigg|^r{\rm{d}}t\bigg)^{\frac{1}{r}}\nonumber\\
=&\sqrt{\frac{1}{2\pi |b|}}\left\|(e^{j\frac{a}{2b}(\cdot)^2} f(\cdot))*(e^{ja(\cdot)+j\frac{a}{2b}(\cdot)^2} g(\cdot+b))\right\|_r\nonumber\\
\triangleq&\sqrt{\frac{1}{2\pi |b|}}\|\tilde{f}* \tilde{g}\|_r,\label{ThmC:2}
\end{align}
\end{small}
where $\tilde{f}(\cdot)\triangleq e^{j\frac{a}{2b}(\cdot)^2} f(\cdot)$,\; $\tilde{g}(\cdot)\triangleq e^{ja(\cdot)+j\frac{a}{2b}(\cdot)^2} g(\cdot+b)$.
Note that $\tilde{f}\in L^p(\mathbb{R}),\; \tilde{g}\in L^q(\mathbb{R})$. Applying the classical Young's inequality (\ref{Prop:A1}) for the functions $\tilde{f}$ and $\tilde{g}$, we have
\begin{equation}\label{ThmC:3}
\|\tilde{f}*\tilde{g}\|_r\le A_pA_qA_{r^{\prime}}\|\tilde{f}\|_p\|\tilde{g}\|_q,
\end{equation}
Note that $\|\tilde{f}\|_p=\|f\|_p,\;\|\tilde{g}\|_q=\|g\|_q$.
Substituting (\ref{ThmC:3}) into (\ref{ThmC:2}), we get
\begin{eqnarray}
\|f\otimes_A g\|_r
&=&\sqrt{\frac{1}{2\pi |b|}}\|\tilde{f}* \tilde{g}\|_r\nonumber\\
&\le&\sqrt{\frac{1}{2\pi |b|}}A_p A_q A_{r^{\prime}}\|\tilde{f}\|_p\|\tilde{g}\|_q\nonumber\\
&=&\sqrt{\frac{1}{2\pi |b|}}A_p A_q A_{r^{\prime}}\|f\|_p\|g\|_q,
\end{eqnarray}
which completes the proof.
\end{proof}

\subsection{Solvability for One Class of Convolution Equations}\label{Sec:Equations}
In this subsection, we mainly discuss the solution for a class of convolution equations associated with the canonical convolution operator $\otimes_A$.
Assume that $\lambda\in \mathbb{C}$, and $f,\;g\in L^1(\mathbb{R})$ are given, $\phi$ is unknown, consider the following canonical convolution equation:
\begin{equation}\label{Eq:Cov1}
\lambda \phi(t)+(g\otimes_A \phi)(t)=f(t).
\end{equation}
In the sequel, we will determine the value of $\phi$ .

Before presenting our main result, we give a lemma, which is very important for proving our theorem.

\begin{Lemma}\label{Eq:L1}
Let $\Lambda(u):=\lambda+\mathcal{L}_{A}g(u)\Phi(u)$, then the following two statements hold:
\begin{enumerate}
  \item[$(a)$]If $\lambda\ne 0$, then there exists a constant C, such that $\Lambda(u)\neq 0$ for every $|u|> C$.
  \item[$(b)$] If for all $u\in \mathbb{R}$, $\Lambda(u)\neq 0$, then $\frac{1}{\Lambda(u)}$ is continuous and bounded on $\mathbb{R}$.
\end{enumerate}
\end{Lemma}

The proof of Lemma~\ref{Eq:L1} is similar to those of \cite[Proposition~7]{ACTT2017B} and \cite[Proposition~1]{ACTT2017A}, hence, we omit the proof.

\begin{Theorem}\label{Thm:B}
Let $\Lambda(u)\neq 0$ for all $u\in \mathbb{R}$. Suppose that one of the following two conditions holds:
\begin{enumerate}\label{ThmB:B1}
  \item [(1)] $\lambda\ne0$, and $\mathcal{L}_{A}f\in L^1(\mathbb{R})$;
  \item [(2)] $\lambda=0$, and $\frac{\mathcal{L}_{A}f}{\mathcal{L}_{A}g}\in L^1(\mathbb{R})$.
\end{enumerate}
Then equation (\ref{Eq:Cov1}) has a solution in $L^1(\mathbb{R})$ if and only if $\mathcal{L}_{A^{-1}}\Big(\frac{\mathcal{L}_{A}f}{\Lambda}\Big)\in L^1(\mathbb{R})$. Furthermore, the solution has the form of $\phi=\mathcal{L}_{A^{-1}}\Big(\frac{\mathcal{L}_{A}f}{\Lambda}\Big)$.
\end{Theorem}

\begin{proof}
We only consider the case when the condition (1) is satisfied.
Since $\Phi(u)=e^{ju-j\frac{d}{2b}u^2-j\frac{ab}{2}}$, $|\Phi(x)|=1$, we know that $\frac{1}{\Phi}$ is continuous and bounded on $\mathbb{R}$. Hence,
$\frac{\mathcal{L}_{A}f}{\mathcal{L}_{A}g}\in L^1(\mathbb{R})$ if and only if $\frac{\mathcal{L}_{A}f}{\Phi\mathcal{L}_{A}g}\in L^1(\mathbb{R})$. Therefore, the case (2) becomes to the case (1).

Necessity: Suppose that equation (\ref{Eq:Cov1}) has a solution $\phi\in L^1(\mathbb{R})$. Multiplying the operator $\mathcal{L}_{A}$ to the both sides of the equation (\ref{Eq:Cov1}), then we have
\begin{equation}\label{ThmB:B2}
\lambda \mathcal{L}_{A}\phi(u)+\mathcal{L}_{A}(g\otimes_A \phi)(u)=\mathcal{L}_{A}f(u).
\end{equation}
By using (\ref{ThmA:2}), we obtain
\begin{equation}\label{ThmB:B2}
\lambda \mathcal{L}_{A}\phi(u)+\Phi(u)\mathcal{L}_{A}g(u) \mathcal{L}_{A}\phi(u)=\mathcal{L}_{A}f(u),
\end{equation}
i.e.,
\begin{equation}\label{ThmB:B3}
(\lambda+\Phi(u)\mathcal{L}_{A}g(u))\mathcal{L}_{A}\phi(u)=\mathcal{L}_{A}f(u).
\end{equation}
Since $\lambda\neq 0$, then
\[
\Lambda(u)=\lambda+\Phi(u)\mathcal{L}_{A}g(u)\neq 0
\]
for all $u\in \mathbb{R}$. Therefore, the equation (\ref{ThmB:B3}) becomes
\begin{equation}\label{ThmB:B4}
\mathcal{L}_{A}\phi(u)=\frac{\mathcal{L}_{A}f(u)}{\Lambda(u)}.
\end{equation}
By Lemma~\ref{Eq:L1}, we know that $\frac{1}{\Lambda(u)}$ is continuous and bounded on $\mathbb{R}$. Since $\mathcal{L}_{A}f\in L^1(\mathbb{R})$, we have $\frac{\mathcal{L}_{A}f(u)}{\Lambda(u)}\in L^1(\mathbb{R})$. Taking inverse LCT transform to the both sides of equation (\ref{ThmB:B4}), we obtain the solution
\[
\phi(t)=\mathcal{L}_{A^{-1}}\Big(\frac{\mathcal{L}_{A}f(u)}{\Lambda(u)}\Big)(t).
\]
Since $\phi\in L^1(\mathbb{R})$, then $\mathcal{L}_{A^{-1}}\Big(\frac{\mathcal{L}_{A}f}{\Lambda}\Big)\in L^1(\mathbb{R}).$

Sufficiency: Let
\[
\phi(t):=\mathcal{L}_{A^{-1}}\Big(\frac{\mathcal{L}_{A}f(u)}{\Lambda(u)}\Big)(t).
 \]
Then, we have $\phi\in L^1(\mathbb{R})$. Applying the LCT to $\phi$, we get
\[
\mathcal{L}_{A}\phi(u)=\frac{\mathcal{L}_{A}f(u)}{\Lambda(u)}.
\]
That is to say,
\[
(\lambda+\Phi(u)\mathcal{L}_{A}g(u))\mathcal{L}_{A}\phi(u)=\mathcal{L}_{A}f(u).
\]
By using (\ref{ThmA:2}) again, we obtain
\[
\mathcal{L}_{A}\left\{\lambda \phi(t)+(g\otimes_A \phi)(t)\right\}(u)=\mathcal{L}_{A}f(u).
\]
Due to the uniqueness of LCT operator $\mathcal{L}_{A}$, $\phi$ satisfies the equation (\ref{Eq:Cov1}) for almost every $t\in \mathbb{R}$, which means that equation (\ref{Eq:Cov1}) has a solution.
This completes the proof.
\end{proof}

\begin{Corollary}
 Let $A=(\cos\alpha,\sin\alpha,-\sin\alpha,\cos\alpha),$ then Theorem~\ref{Thm:B} reduces to the Theorem~3 mentioned in \cite{ACTT2017A}.
\end{Corollary}

\begin{Remark}
Similar to the Definition~\ref{def:convo}, we can define another canonical convolution operator $\odot_A$ by
\begin{eqnarray}
\Big(f\odot_A g\Big)(t)&=&\sqrt{\frac{1}{j2\pi b}}\int_{-\infty}^{+\infty}f(u)g(t-u-b){}\nonumber\\
{}&&\times e^{j\frac{a}{b}u^{2}-j\frac{a}{b}ut-jat+jau}{\rm{d}}u.\label{def:Rem1}
\end{eqnarray}
Then, the new canonical convolution operator $\odot_A$ also has three properties: Commutative property, associative property, and distributive property. In addition, the statements in Theorem~\ref{Thm:A}, Theorem~\ref{Thm:C}, and Theorem~\ref{Thm:B} also hold for operator $\odot_A$ with some minor adjustments. Due to the similarity, we omit the proof of this claim.
\end{Remark}

\section{Conclusion}\label{Sec:Conclusion}
In this paper, we first define a new canonical convolution operator, which is much more flexible and simple than the existing ones. Then, we show that it satisfies the generalized convolution theorem and Young's inequality. Finally, we investigate the solvability of a class of convolution equations associated with the new canonical convolution operator.

\section*{Acknowledgements}
The author thanks the referees very much for carefully reading the paper and for elaborate and valuable suggestions.

\end{document}